\documentclass[10pt,letterpaper]{amsart}
\usepackage{amsmath,amssymb,amsfonts,amsthm}
\usepackage{graphicx}
\usepackage{bm}
\usepackage{mathrsfs}
\usepackage{hyperref}
\usepackage{cleveref}
\usepackage{doi}

\numberwithin{equation}{section}

\newtheorem{theorem}{Theorem}[section]
\newtheorem{proposition}[theorem]{Proposition}
\newtheorem{lemma}[theorem]{Lemma}
\newtheorem{corollary}[theorem]{Corollary}
\theoremstyle{definition}
\newtheorem{definition}[theorem]{Definition}
\newtheorem{remark}[theorem]{Remark}
\newtheorem{example}[theorem]{Example}
\newtheorem{assumption}[theorem]{Assumption}

\title[Rigidity and Holonomy of Time-Scaled Intertwining Cocycles]
{Gauge Rigidity and Holonomy Classification of Time-Scaled Intertwining Cocycles on Hilbert Bundles}

\author{Anton Alexa}
\address{Independent Researcher, Chernivtsi, Ukraine}
\email{mail@antonalexa.com}


\begin{document}

\begin{abstract}
We extend the discrete gauge rigidity of time-scaled intertwining cocycles
to operator networks over a connected manifold $M$ and classify the
holonomy that the continuous setting admits. For dissipative semigroups
$\mathcal{S}_x(t)=e^{-tA_x}$ linked by transport operators
$K_\gamma\mathcal{S}_{\gamma(0)}(t)=\mathcal{S}_{\gamma(1)}(\lambda(\gamma)t)K_\gamma$
along paths $\gamma$, we prove that the scaling field carries no monodromy:
$\lambda(x,y)=\tau(x)/\tau(y)$ for a continuous positive function $\tau$,
unique up to a multiplicative constant, with multiplicativity forced by
the intertwining relation rather than assumed. For regular cocycles the
transport defines a bounded operator connection on a Hilbert bundle over
$M$; the holonomy of a closed loop is
confined to the commutant of the generator and decomposes into spectral
sectors, on each of which it is adiabatic (Berry--Wilczek--Zee) transport
twisted by an independent commutant-valued sector potential; for flat
unitary cocycles the holonomy is classified by unitary representations
of $\pi_1(M)$ in the commutant. An explicit rotating network over $S^1$
exhibits holonomy equal to the parity operator while the Berry connection
one-form of every level vanishes identically: the monodromy is carried by
$\mathbb{Z}_2$ orientation classes of M\"{o}bius eigenline bundles.
\end{abstract}
\maketitle

\section{Introduction}
\label{sec:introduction}

Intertwining relations between operator semigroups are classical in functional
analysis~\cite{hille-phillips,pazy1983,engel-nagel}, but most existing
frameworks treat either a single pair of generators or time-preserving transfer
maps~\cite{haase2006}. In a companion paper~\cite{alexa2026} we introduced a network formalism
for dissipative semigroups $\mathcal{S}_i(t)=e^{-tA_i}$ indexed by a finite
set $I$, linked by time-scaled intertwining cocycles
$K_{ij}\mathcal{S}_j(t)=\mathcal{S}_i(\lambda_{ij}t)K_{ij}$ with positive
scaling factors $\lambda_{ij}$. The central result of that theory is a gauge
rigidity theorem: the scaling factors are necessarily of the form
$\lambda_{ij}=\tau_i/\tau_j$ for a single family of positive constants
$\{\tau_i\}$, and the renormalized generators $\{\tau_i A_i\}$ form a common
isospectral class. In the discrete setting, all cycle products of transport
operators are trivial --- the absence of continuous paths makes flatness
automatic rather than a constraint.

The natural question is what happens when the index set is replaced by a
continuous parameter space. In this paper we answer that question. Let $M$ be
a connected smooth manifold, $\{A_x\}_{x\in M}$ a smooth family of positive
self-adjoint operators with compact resolvent on a separable Hilbert space
$\mathcal{H}$, and $\{K_{x,y}\}$ a continuous family of bounded injective
operators satisfying the cocycle identity $K_{x,z}=K_{x,y}K_{y,z}$ and the
scaled intertwining relation
$K_{x,y}\mathcal{S}_y(t)=\mathcal{S}_x(\lambda(x,y)t)K_{x,y}$
for a continuous scaling field $\lambda\colon M\times M\to\mathbb{R}_{>0}$.
Our first main result shows that gauge rigidity persists in this
continuous setting.

\begin{theorem}[Continuous Gauge Rigidity]
\label{thm:intro-gauge}
There exists a continuous function $\tau\colon M\to\mathbb{R}_{>0}$,
unique up to a positive multiplicative constant, such that
\begin{equation}
\lambda(x,y) = \frac{\tau(x)}{\tau(y)}
\qquad\text{for all }x,y\in M.
\end{equation}
Consequently, all renormalized generators $\tau(x)A_x$ belong to a common
isospectral class: $\sigma(\tau(x)A_x)=\sigma(\tau(y)A_y)$ for all $x,y\in M$.
\end{theorem}

The key point --- and what makes this non-trivial --- is that multiplicativity
of $\lambda$, i.e.\ $\lambda(x,z)=\lambda(x,y)\lambda(y,z)$, is not assumed
but derived from the semigroup intertwining relation. The logical chain is:
semigroup intertwining forces generator transport
$K_{x,y}A_y=\lambda(x,y)A_xK_{x,y}$, which in turn forces multiplicativity
of $\lambda$, which then yields the gauge form by a standard cocycle argument.

In the continuous setting a new geometric phenomenon appears that is absent
in the discrete case. Transport around a non-contractible loop need not
return to the identity, so the natural general formulation indexes the
transfer operators by paths rather than by pairs of endpoints; a two-point
cocycle $K_{x,y}$ is precisely the special case of path-independent
transport (Section~\ref{sec:connection}). For regular cocycles the
path transport defines a bounded operator connection on a Hilbert
bundle
$\mathcal{H}\times M\to M$~\cite{brody-hughston2001}, and for a closed
loop $\gamma$ in $M$ the holonomy operator
$K_\gamma\in\mathcal{B}(\mathcal{H})$ need not be the identity. Two
structural constraints govern this freedom. First, the scalar layer is
topologically inert: the scaling field never acquires monodromy, so the
gauge function $\tau$ is single-valued on any base. Second, the holonomy
of a closed loop commutes with the semigroup at the base point, hence
preserves every spectral subspace of the generator. The two constraints
combine into the second main result.

\begin{theorem}[Holonomy classification; informal statement]
\label{thm:intro-classification}
Let $\{K_\gamma,\lambda\}$ be a path intertwining cocycle
\textup{(Definition~\ref{def:path-cocycle})} on a continuous operator
network over $M$.
\emph{(a)} The scaling field satisfies $\lambda(\gamma)=1$ for every
closed loop $\gamma$: time-scaling never acquires monodromy
\textup{(Lemma~\ref{lem:no-scalar-monodromy})}.
\emph{(b)} For flat unitary cocycles, the holonomy defines a group
homomorphism $\rho\colon\pi_1(M,x_0)\to\mathcal{U}(\{A_{x_0}\}')$ into
the unitary commutant of the generator, and every such homomorphism is
realized on an isospectral network --- over $S^1$, even on a constant
one \textup{(Theorem~\ref{thm:classification})}.
\end{theorem}

Nontrivial monodromy is a concrete phenomenon, not a formal
possibility: an explicit rotating network on $\ell^2$ over $S^1$
exhibits holonomy equal to the parity operator
(Example~\ref{ex:holonomy}), with a rotating anisotropic oscillator
on $L^2(\mathbb{R}^2)$ as its physical counterpart
(Remark~\ref{rem:oscillator}).

Finally, we make the relation to the Berry connection of adiabatic
quantum mechanics~\cite{berry1984,simon1983} precise. The transport
preserves each renormalized spectral sector
$E_\mu(x)=\ker(\tau(x)A_x-\mu I)$, and its restriction
decomposes canonically: it is Berry parallel transport --- in the
degenerate case, the non-Abelian transport of Wilczek and
Zee~\cite{wilczek-zee1984} --- twisted by an independent
commutant-valued \emph{sector potential}, the spectral-sector block
of the connection form (Theorem~\ref{prop:berry}). Adiabatic theory is thus
a finite-rank sector of the present framework, and it captures the
transport exactly on the class of adiabatically normalized cocycles,
for which the sector potentials vanish. In particular the holonomy is
not determined by the Berry connection one-forms
alone~\cite{carollo2020}, in two distinct ways: through the sector
potential, and --- even for normalized cocycles --- through global
topology: in the example above the Berry one-form of every level
vanishes identically, the eigenfunctions being real, yet the holonomy
acts as $(-1)^{n}$ on the $n$-th eigenline. The latter monodromy is
carried by $\mathbb{Z}_2$ orientation classes of M\"{o}bius eigenline
bundles, the infinite-dimensional counterpart of the Longuet--Higgins
sign change~\cite{herzberg-longuet-higgins1963}, and is invisible to
the local connection data.
\section{Continuous Operator Networks}
\label{sec:networks}

Throughout this paper $M$ denotes a connected smooth manifold and
$\mathcal{H}$ a fixed separable Hilbert space. We work with a trivial
Hilbert bundle $\mathcal{H}\times M\to M$; the fiber over each $x\in M$
is identified with $\mathcal{H}$. The following regularity assumption
governs the dependence of the operators on the parameter $x$.

\begin{assumption}[Common domain regularity]\label{ass:R}
There exists a dense subspace $D\subset\mathcal{H}$, independent of $x$,
such that $\mathcal{D}(A_x)=D$ for all $x\in M$, and the map
$x\mapsto A_x$ is continuous in the strong resolvent sense: for some
(hence every) $\mu$ in the resolvent set,
\begin{equation}
(A_x-\mu)^{-1}\to (A_y-\mu)^{-1}
\quad\text{strongly in }\mathcal{H}
\quad\text{as }x\to y.
\end{equation}
\end{assumption}

Assumption~\ref{ass:R} is satisfied by analytic families of operators
in the sense of Kato~\cite{kato1980}, and by Schr\"{o}dinger operators
with smoothly varying potentials~\cite{reed-simon2}. The case
of $x$-dependent domains is handled by the Kato analytic family
framework~\cite[Ch.~VII]{kato1980} and is left for future work.

\begin{definition}[Continuous operator network]
\label{def:network}
A continuous operator network over $M$ consists of a family
$\{A_x\}_{x\in M}$ of positive self-adjoint operators on $\mathcal{H}$
with compact resolvent, satisfying Assumption~\ref{ass:R}, together with the
associated strongly continuous contraction semigroups
$\mathcal{S}_x(t)=e^{-tA_x}$, $t\ge 0$ \cite{pazy1983,engel-nagel}.
\end{definition}

Since each $A_x$ has compact resolvent, its spectrum is purely discrete,
$\sigma(A_x)=\{\alpha_{x,1}\le\alpha_{x,2}\le\cdots\}\nearrow\infty$
with finite multiplicities, and $\mathcal{S}_x(t)$ is compact for
every $t>0$~\cite{reed-simon2}; moreover
$\|\mathcal{S}_x(t)\|_{\mathcal{B}(\mathcal{H})}\le e^{-\alpha_{x,1}t}$.

\begin{remark}[Regularity for the connection]
Assumption~\ref{ass:R} is sufficient for the gauge rigidity results of
Sections~\ref{sec:cocycles}--\ref{sec:gauge}. The construction of the
infinitesimal operator connection in Section~\ref{sec:connection} requires
the stronger condition that the cocycle admits a bounded connection
one-form; this regularity is stated as Assumption~\ref{ass:D} in
that section.
\end{remark}

\begin{remark}[Unitary networks]
\label{rem:unitary}
All results of this paper hold, with identical proofs, when the
contraction semigroups $\mathcal{S}_x(t)=e^{-tA_x}$, $t\ge0$, are
replaced by the unitary groups $U_x(t)=e^{-itA_x}$, $t\in\mathbb{R}$,
generated by the same positive self-adjoint operators: every argument
below uses only the generator relation obtained by differentiating the
intertwining identity at $t=0$, together with discreteness of the
spectrum and the existence of a strictly positive eigenvalue. Two
points deserve note. First, the time-scaled intertwining relation is
not invariant under the shift $A_x\mapsto A_x+c$, so the normalization
of the ground energy is part of the data; positivity of $A_x$ fixes
it. Second, in the unitary reading the results apply directly to
Schr\"{o}dinger dynamics, which is the natural setting both for the
adiabatic interpretation of Section~\ref{sec:berry} and for the
physical reading of Remark~\ref{rem:second-clock}.
\end{remark}

\section{Two-Point Intertwining Cocycles}
\label{sec:cocycles}

With the operator network in place, we introduce the intertwining
structure. The central object is a family of bounded operators
$K_{x,y}$ linking the semigroups at different base points, subject to
a cocycle identity~\eqref{eq:cocycle-C} and a time-scaled intertwining
relation~\eqref{eq:cocycle-I}. The scaling field $\lambda(x,y)$ is
part of the data, but its multiplicativity --- the key property needed
for gauge rigidity --- is not assumed: it is a consequence of the
semigroup intertwining and is derived in
Proposition~\ref{prop:multiplicativity} below.

\begin{definition}[Two-point intertwining cocycle]
\label{def:cocycle}
A \emph{two-point intertwining cocycle} on the operator network
$\{A_x\}_{x\in M}$ is a family $\{K_{x,y}\}_{x,y\in M}$ of bounded
operators on $\mathcal{H}$, together with a continuous function
$\lambda\colon M\times M\to\mathbb{R}_{>0}$, satisfying the following
conditions for all $x,y,z\in M$ and all $t\ge 0$:
\begin{align}
\label{eq:cocycle-C}
&K_{x,z} = K_{x,y}K_{y,z} \tag{C}\\[4pt]
\label{eq:cocycle-I}
&K_{x,y}\,\mathcal{S}_y(t) = \mathcal{S}_x\!\left(\lambda(x,y)\,t\right)K_{x,y} \tag{I}
\end{align}
and the map $(x,y)\mapsto K_{x,y}\,h$ is continuous in $\mathcal{H}$ for
every $h\in\mathcal{H}$ (joint strong continuity). Each operator
$K_{x,y}$ is assumed to be injective. The
function $\lambda(x,y)$ is called the \emph{scaling field} of the
cocycle.
\end{definition}

\begin{remark}[Discrete case]
\label{rem:discrete-case}
When $M = I$ is a finite set with the discrete topology,
Definition~\ref{def:cocycle} reduces exactly to the time-scaled
intertwining cocycle of the companion paper \cite{alexa2026}, and
condition~\eqref{eq:cocycle-C} is the cycle consistency identity of
that theory. All results of the present section generalize the
corresponding statements of the companion paper from finite $I$ to
continuous $M$; the proofs are structurally identical.
\end{remark}

\begin{proposition}[Diagonal and invertibility]
\label{prop:diagonal}
Let $\{K_{x,y},\lambda\}$ be a two-point intertwining cocycle. Then:
\begin{enumerate}
\item[\emph{(i)}] $K_{x,x} = I_{\mathcal{H}}$ for every $x\in M$.
\item[\emph{(ii)}] Each $K_{x,y}$ is boundedly invertible, with
$K_{x,y}^{-1} = K_{y,x}$.
\item[\emph{(iii)}] $\lambda(x,x) = 1$ and
$\lambda(x,y)\lambda(y,x) = 1$ for all $x,y\in M$.
\end{enumerate}
\end{proposition}

\begin{proof}
\emph{(i).} Setting $z = x = y$ in \eqref{eq:cocycle-C} gives
$K_{x,x} = K_{x,x}K_{x,x}$, so $K_{x,x}$ is an idempotent on
$\mathcal{H}$. Since $K_{x,x}$ is injective, the relation
$K_{x,x}(I-K_{x,x}) = 0$ forces $I - K_{x,x} = 0$, i.e.\
$K_{x,x} = I$.

\emph{(ii).} From \eqref{eq:cocycle-C} with $z$ replaced by $x$:
$K_{x,x} = K_{x,y}K_{y,x}$, whence $K_{x,y}K_{y,x} = I$
by part~(i). Similarly, setting $x$ replaced by $y$:
$K_{y,x}K_{x,y} = K_{y,y} = I$. Thus $K_{y,x}$ is both a left and
right inverse of $K_{x,y}$, so $K_{x,y}^{-1} = K_{y,x}$.

\emph{(iii).} Setting $x = y$ in \eqref{eq:cocycle-I} and using part~(i)
gives $\mathcal{S}_x(t) = \mathcal{S}_x(\lambda(x,x)t)$ for all $t\ge0$.
Since $\sigma(A_x)$ is discrete and unbounded, $A_x$ has an eigenvector
$\phi$ with eigenvalue $\alpha>0$; applying the identity to $\phi$ gives
$e^{-\alpha t}=e^{-\alpha\lambda(x,x)t}$ for all $t$, so
$\lambda(x,x)=1$. Next, composing \eqref{eq:cocycle-I} for $K_{x,y}$ and
$K_{y,x}$ and using $K_{x,y}K_{y,x}=I$ from part~(ii),
\begin{equation*}
\mathcal{S}_x(t)=K_{x,y}K_{y,x}\,\mathcal{S}_x(t)
=K_{x,y}\,\mathcal{S}_y(\lambda(y,x)t)\,K_{y,x}
=\mathcal{S}_x\!\left(\lambda(x,y)\lambda(y,x)t\right),
\end{equation*}
and the same eigenvector argument yields
$\lambda(x,y)\lambda(y,x) = 1$.
\end{proof}

\begin{proposition}[Generator intertwining]
\label{prop:generator}
Under Assumption~\ref{ass:R}, for every $x,y\in M$ and every
$h\in D$,
\begin{equation}
\label{eq:generator-intertwining}
K_{x,y}\,A_y\,h = \lambda(x,y)\,A_x\,K_{x,y}\,h.
\end{equation}
\end{proposition}

\begin{proof}
Fix $x,y\in M$ and $h\in D$. By \eqref{eq:cocycle-I}, for $t>0$,
\begin{equation*}
\frac{\mathcal{S}_x(\lambda(x,y)t)\,K_{x,y}h-K_{x,y}h}{t}
= K_{x,y}\,\frac{\mathcal{S}_y(t)h-h}{t}.
\end{equation*}
Since $h\in\mathcal{D}(A_y)$ and $K_{x,y}$ is bounded, the right-hand
side converges to $-K_{x,y}A_yh$ as $t\downarrow0$. Hence the limit on
the left exists, and by the characterization of the generator of a
$C_0$-semigroup this means
$K_{x,y}h\in\mathcal{D}(A_x)$ with
$-\lambda(x,y)A_xK_{x,y}h=-K_{x,y}A_yh$, which is
\eqref{eq:generator-intertwining}. Note that membership
$K_{x,y}h\in\mathcal{D}(A_x)$ is a conclusion, not a hypothesis.
Applying the same argument to $K_{y,x}=K_{x,y}^{-1}$ yields
$K_{x,y}\,\mathcal{D}(A_y)=\mathcal{D}(A_x)$.
\end{proof}

\begin{proposition}[Multiplicativity of the scaling field]
\label{prop:multiplicativity}
For every $x,y,z\in M$,
\begin{equation}
\label{eq:multiplicativity}
\lambda(x,z) = \lambda(x,y)\,\lambda(y,z).
\end{equation}
\end{proposition}

\begin{proof}
This is the key step; multiplicativity is not assumed but derived
from the semigroup structure. Apply the intertwining
relation~\eqref{eq:cocycle-I} twice:
\begin{align*}
K_{x,z}\,\mathcal{S}_z(t)
&= \mathcal{S}_x\!\left(\lambda(x,z)\,t\right)K_{x,z}.
\end{align*}
On the other hand, using \eqref{eq:cocycle-C} and then
\eqref{eq:cocycle-I} at $(y,z)$ and at $(x,y)$ in succession:
\begin{align*}
K_{x,z}\,\mathcal{S}_z(t)
&= K_{x,y}\,K_{y,z}\,\mathcal{S}_z(t) \\
&= K_{x,y}\,\mathcal{S}_y\!\left(\lambda(y,z)\,t\right)K_{y,z} \\
&= \mathcal{S}_x\!\left(\lambda(x,y)\lambda(y,z)\,t\right)K_{x,y}\,K_{y,z} \\
&= \mathcal{S}_x\!\left(\lambda(x,y)\lambda(y,z)\,t\right)K_{x,z}.
\end{align*}
Comparing the two expressions:
$\mathcal{S}_x\!\left(\lambda(x,z)\,t\right)K_{x,z}\,h
= \mathcal{S}_x\!\left(\lambda(x,y)\lambda(y,z)\,t\right)K_{x,z}\,h$
for all $h\in\mathcal{H}$ and all $t\ge 0$. By
Proposition~\ref{prop:diagonal}(ii), $K_{x,z}$ is boundedly
invertible, hence surjective, so
$\mathcal{S}_x(\lambda(x,z)t)=\mathcal{S}_x(\lambda(x,y)\lambda(y,z)t)$
on all of $\mathcal{H}$. Applying this identity to an eigenvector of
$A_x$ with eigenvalue $\alpha>0$ (which exists since $\sigma(A_x)$ is
discrete and unbounded) gives
$e^{-\alpha\lambda(x,z)t}=e^{-\alpha\lambda(x,y)\lambda(y,z)t}$ for all
$t\ge0$, hence $\lambda(x,z) = \lambda(x,y)\lambda(y,z)$.
\end{proof}

\section{Continuous Gauge Rigidity}
\label{sec:gauge}

The multiplicativity of $\lambda$ established in
Proposition~\ref{prop:multiplicativity} is the only input needed to
derive the gauge form: fix any base point, define $\tau$ by
restriction to that base point, and recover the full field by
multiplicativity. No topological input is required --- for a
two-point cocycle, defined on all pairs, both existence and
uniqueness of $\tau$ are purely algebraic.

\begin{theorem}[Continuous Gauge Rigidity]
\label{thm:continuous-gauge}
Let $\{A_x\}_{x\in M}$ be a continuous operator network and let
$\{K_{x,y},\lambda\}$ be a two-point intertwining cocycle. Then there
exists a continuous function $\tau\colon M\to\mathbb{R}_{>0}$ such that
\begin{equation}
  \label{eq:gauge}
  \lambda(x,y)=\frac{\tau(x)}{\tau(y)}
  \qquad\text{for all }x,y\in M.
\end{equation}
The function $\tau$ is unique up to a positive multiplicative constant.
\end{theorem}

\begin{proof}
Fix a base point $x_0\in M$ and define $\tau\colon M\to\mathbb{R}_{>0}$
by $\tau(x) := \lambda(x,x_0)$. This is positive by definition and
continuous because $\lambda$ is continuous.

\emph{Gauge form.} For any $x,y\in M$, apply
Proposition~\ref{prop:multiplicativity} with the triple $(x,y,x_0)$:
\begin{equation*}
\lambda(x,x_0) = \lambda(x,y)\,\lambda(y,x_0),
\end{equation*}
i.e.\ $\tau(x) = \lambda(x,y)\,\tau(y)$. Since $\tau(y)>0$ this gives
$\lambda(x,y) = \tau(x)/\tau(y)$.

\emph{Uniqueness.} Suppose $\tau'\colon M\to\mathbb{R}_{>0}$ also
satisfies $\lambda(x,y)=\tau'(x)/\tau'(y)$. Then
$\tau'(x)/\tau'(y) = \tau(x)/\tau(y)$, i.e.\
$(\tau'/\tau)(x) = (\tau'/\tau)(y)$, for \emph{all} pairs
$x,y\in M$; hence $\tau'/\tau$ is constant. Connectedness of $M$
plays no role here.
\end{proof}

\begin{corollary}[Continuous spectral rigidity]
\label{cor:spectral-rigidity}
Under the hypotheses of Theorem~\ref{thm:continuous-gauge},
\begin{equation}
\label{eq:spectral-rigidity}
\sigma\!\left(\tau(x)A_x\right) = \sigma\!\left(\tau(y)A_y\right)
\qquad\text{for all }x,y\in M.
\end{equation}
That is, all renormalized generators $\tau(x)A_x$ belong to a common
isospectral class.
\end{corollary}

\begin{proof}
By Proposition~\ref{prop:diagonal}(ii), $K_{x,y}$ is boundedly
invertible. By Proposition~\ref{prop:generator},
$K_{x,y}A_y = \lambda(x,y)A_xK_{x,y}$ on $D$. Multiplying both sides
by $\tau(y)$ and using $\lambda(x,y)=\tau(x)/\tau(y)$:
\begin{equation*}
K_{x,y}\,(\tau(y)A_y) = \tau(x)A_x\,K_{x,y}.
\end{equation*}
Thus $\tau(y)A_y$ and $\tau(x)A_x$ are similar via the boundedly
invertible operator $K_{x,y}$; by
Proposition~\ref{prop:generator} the similarity also matches the
domains, $K_{x,y}\mathcal{D}(A_y)=\mathcal{D}(A_x)$, so the two
unbounded operators are similar in the precise sense and have equal
spectra.
\end{proof}

\begin{remark}[The logical chain]
\label{rem:chain}
The gauge form \eqref{eq:gauge} is algebraically elementary once
multiplicativity is in hand. The non-trivial content is that
multiplicativity is forced by the semigroup structure rather
than imposed as a hypothesis. The full logical chain is:
\begin{align*}
&\text{semigroup intertwining } \eqref{eq:cocycle-I}\\
&\Rightarrow\quad
K_{x,y}A_y = \lambda(x,y)A_xK_{x,y}
\quad\text{(Proposition~\ref{prop:generator})}\\
&\Rightarrow\quad
\lambda(x,z)=\lambda(x,y)\lambda(y,z)
\quad\text{(Proposition~\ref{prop:multiplicativity})}\\
&\Rightarrow\quad
\lambda(x,y)=\tau(x)/\tau(y)
\quad\text{(Theorem~\ref{thm:continuous-gauge}).}
\end{align*}
Removing any step breaks the argument.
\end{remark}

\section{Path Cocycles and the Operator Connection}
\label{sec:connection}

Theorem~\ref{thm:continuous-gauge} and its corollary are obtained under
Assumption~\ref{ass:R} alone, for transport indexed by pairs of points.
A two-point cocycle satisfying \eqref{eq:cocycle-C} is, by construction,
path-independent: the transport from $y$ to $x$ does not remember how
the parameter moved. For flat regular cocycles over a simply connected
base this is the general situation
(Theorem~\ref{thm:flatness}); in general, however, the two-point
format excludes precisely the phenomena we wish to study ---
curvature, and monodromy around non-contractible loops. Indeed, for a two-point cocycle every loop
transport equals $K_{x_0,x_0}=I$ by
Proposition~\ref{prop:diagonal}\emph{(i)}. The natural general object
indexes transport by paths.

\begin{definition}[Path intertwining cocycle]\label{def:path-cocycle}
A \emph{path intertwining cocycle} on the operator network
$\{A_x\}_{x\in M}$ assigns to every piecewise smooth path
$\gamma\colon[0,1]\to M$ a bounded injective operator
$K_\gamma\in\mathcal{B}(\mathcal{H})$ and a number $\lambda(\gamma)>0$
such that for all admissible paths and all $t\ge0$:
\begin{enumerate}
\item[\emph{(P1)}] $K_\gamma$ is invariant under orientation-preserving
reparametrization; $K_\gamma=I$ for constant $\gamma$; and
$K_{\bar\gamma}=K_\gamma^{-1}$, where $\bar\gamma$ denotes the reversed
path.
\item[\emph{(P2)}] $K_{\gamma_2\cdot\gamma_1}=K_{\gamma_2}K_{\gamma_1}$
whenever the concatenation $\gamma_2\cdot\gamma_1$ is defined.
\item[\emph{(P3)}] $K_\gamma\,\mathcal{S}_{\gamma(0)}(t)
=\mathcal{S}_{\gamma(1)}\!\left(\lambda(\gamma)\,t\right)K_\gamma$.
\item[\emph{(P4)}] $s\mapsto K_{\gamma|_{[0,s]}}\,h$ is continuous in
$\mathcal{H}$ for every $h\in\mathcal{H}$.
\end{enumerate}
\end{definition}

\begin{remark}[Relation to two-point cocycles]\label{rem:two-point}
If $\{K_{x,y},\lambda\}$ is a cocycle in the sense of
Definition~\ref{def:cocycle}, then $K_\gamma:=K_{\gamma(1),\gamma(0)}$
and $\lambda(\gamma):=\lambda(\gamma(1),\gamma(0))$ define a path
cocycle whose transport depends only on endpoints; conversely, a path
cocycle with path-independent transport descends to a two-point
cocycle. Definition~\ref{def:path-cocycle} therefore strictly
generalizes Definition~\ref{def:cocycle}, and the generalization is
exactly what admits holonomy. All results of
Sections~\ref{sec:cocycles}--\ref{sec:gauge} extend verbatim to path
cocycles: their proofs use only \emph{(P1)}--\emph{(P3)} applied along
fixed paths. In particular, differentiating \emph{(P3)} at $t=0$ as in
Proposition~\ref{prop:generator} gives the generator transport
$K_\gamma A_{\gamma(0)}h=\lambda(\gamma)A_{\gamma(1)}K_\gamma h$ for
$h\in D$.
\end{remark}

The first structural fact about path cocycles is that the scalar layer
never acquires monodromy: the scaling field cannot detect the topology
of the base.

\begin{lemma}[No scalar monodromy]\label{lem:no-scalar-monodromy}
Let $\{K_\gamma,\lambda\}$ be a path intertwining cocycle. Then
$\lambda(\gamma)=1$ for every closed loop $\gamma$, and
$\lambda(\gamma)$ depends only on the endpoints of $\gamma$.
Consequently there exists $\tau\colon M\to\mathbb{R}_{>0}$, unique up to
a positive multiplicative constant, with
\begin{equation}\label{eq:path-gauge}
\lambda(\gamma)=\frac{\tau(\gamma(1))}{\tau(\gamma(0))}
\qquad\text{for every piecewise smooth path }\gamma.
\end{equation}
\end{lemma}

\begin{proof}
Composing \emph{(P3)} for $\gamma_1$ and $\gamma_2$ and comparing with
\emph{(P3)} for the concatenation, an application to an eigenvector
with positive eigenvalue (as in
Proposition~\ref{prop:multiplicativity}) gives
$\lambda(\gamma_2\cdot\gamma_1)=\lambda(\gamma_2)\lambda(\gamma_1)$.
For a closed loop $\gamma$ at $x_0$, \emph{(P1)} makes $K_\gamma$
invertible, and \emph{(P3)} gives
$K_\gamma\,\mathcal{S}_{x_0}(t)\,K_\gamma^{-1}
=\mathcal{S}_{x_0}(\lambda(\gamma)t)$, hence
$\sigma(A_{x_0})=\lambda(\gamma)\,\sigma(A_{x_0})$ as sets. The set
$\sigma(A_{x_0})$ is discrete and possesses a smallest strictly
positive element $\alpha_\ast$ (compact resolvent,
$\sigma(A_{x_0})\nearrow\infty$). Write $\lambda=\lambda(\gamma)$.
If $\lambda>1$, then $\alpha_\ast\in\sigma(A_{x_0})
=\lambda\,\sigma(A_{x_0})$ gives
$\alpha_\ast/\lambda\in\sigma(A_{x_0})$, a strictly smaller positive
element --- contradiction. If $\lambda<1$, then
$\lambda\alpha_\ast\in\lambda\,\sigma(A_{x_0})=\sigma(A_{x_0})$ is a
strictly smaller positive element --- contradiction. Hence
$\lambda(\gamma)=1$. If
$\gamma_1,\gamma_2$ share endpoints, then
$\bar\gamma_2\cdot\gamma_1$ is a closed loop, so
$\lambda(\gamma_1)/\lambda(\gamma_2)
=\lambda(\bar\gamma_2\cdot\gamma_1)=1$: endpoint dependence follows.
The gauge form \eqref{eq:path-gauge} and uniqueness are then obtained
exactly as in Theorem~\ref{thm:continuous-gauge}, defining $\tau$
through transport from a base point along arbitrary paths (any two
choices give the same value).
\end{proof}

\begin{remark}[Two layers]\label{rem:two-layers}
Lemma~\ref{lem:no-scalar-monodromy} splits the theory into two layers
with opposite behavior: the time-scaling layer is topologically inert
--- pure gauge on any base --- while all geometric information is
carried by the operator layer $\{K_\gamma\}$, whose loop values are
constrained in Section~\ref{sec:holonomy}. The passage from loop
triviality to the coboundary form is the elementary step, as in the
cocycle rigidity of Liv\v{s}ic theory in hyperbolic
dynamics~\cite{livsic1972}; the content here is that loop triviality
itself is not assumed as periodic data but forced by the spectrum.
\end{remark}

\begin{remark}[Physical reading: the second clock effect]
\label{rem:second-clock}
Lemma~\ref{lem:no-scalar-monodromy} admits a physical reading. In
Weyl's scale-gauge theory of 1918~\cite{weyl1918}, the rate of a clock
transported around a closed loop could depend on the loop --- the
\emph{second clock effect} --- and Einstein's objection, appended to
Weyl's paper, was that atoms with different histories would then
display shifted spectral lines, contradicting their observed
sharpness. Within the present model class the lemma turns this
objection into a theorem: for networks of semigroup dynamics with
discrete spectra bounded below, a path-dependent time rate is not
merely unobserved but impossible, and the mechanism is precisely the
existence of a smallest positive spectral level. A detailed
discussion, including the unitary case and the failure of rigidity
for scale-invariant spectra without a spectral bottom, is deferred to
a companion note.
\end{remark}

To pass from parallel transport to an infinitesimal connection --- a
$\mathcal{B}(\mathcal{H})$-valued $1$-form on $M$ --- more than
differentiability along individual curves is needed: a path cocycle
satisfying \emph{(P1)}--\emph{(P4)} may have transport whose first
variation depends on higher-order data of the path (for instance,
$K_\gamma=e^{i\int_\gamma\kappa\,ds}I$ on a surface, with $\kappa$
the signed geodesic curvature, satisfies all four axioms, yet its
derivative along $\gamma$ involves the second-order jet). We
therefore take the existence of a connection form as part of the
regularity, in the spirit of transport functors.

\begin{assumption}[Regular cocycle]\label{ass:D}
There exists a $1$-form $\omega$ on $M$ with values in
$\mathcal{B}(\mathcal{H})$, linear in the tangent vector and
norm-continuous in the base point, such that for every smooth curve
$\gamma\colon[0,1]\to M$ and every $h\in\mathcal{H}$ the map
$P(s):=K_{\gamma|_{[0,s]}}$ satisfies
\begin{equation}
\label{eq:pt-ode}
\frac{d}{ds}P(s)\,h = -\omega_{\gamma(s)}\!\left(\gamma'(s)\right)
P(s)\,h,
\qquad P(0)=I.
\end{equation}
A path cocycle admitting such an $\omega$ is called \emph{regular},
and $\omega$ is its \emph{connection $1$-form}.
\end{assumption}

\begin{remark}[Uniqueness and reconstruction]
\label{def:connection}
The form $\omega$ is uniquely determined by the cocycle:
evaluating \eqref{eq:pt-ode} at $s=0$ gives
$\omega_x(v)h=-\frac{d}{ds}\big|_{0}K_{\gamma|_{[0,s]}}h$ for any
curve with $\gamma(0)=x$, $\gamma'(0)=v$. Conversely, the cocycle is
recovered from $\omega$: since $\omega$ is norm-continuous, the
initial-value problem \eqref{eq:pt-ode} has a unique solution given
by the norm-convergent Dyson series, and by the Gronwall inequality
$K_\gamma=P(1)$ is the unique bounded solution. In particular
$\|\omega_x(v)\|\le C_K\|v\|$ uniformly on compact subsets of $M$.
Regularity holds, for example, whenever the cocycle is generated by
conjugation along the flow of a norm-continuous family of bounded
skew-adjoint generators, as in the example of
Section~\ref{sec:example}; unbounded transport generators (such as
rotations on $L^2(\mathbb{R}^2)$, whose generator is the angular
momentum operator) fall outside this bounded regularity class and
are discussed separately in Remark~\ref{rem:oscillator}.
\end{remark}

\begin{proposition}[Parallel transport from the connection]
\label{prop:parallel-transport}
Under Assumption~\ref{ass:D}, for any smooth curve
$\gamma\colon[0,1]\to M$ the transport operator $K_\gamma$ is the
unique bounded solution of \eqref{eq:pt-ode} evaluated at $s=1$; in
particular the transport is uniquely determined by $\omega$ and the
homotopy-theoretic properties of the cocycle are encoded in
$\omega$.
\end{proposition}

\begin{proof}
Existence is Assumption~\ref{ass:D} itself; uniqueness follows from
the Gronwall inequality in the strong topology, using the uniform
bound $\|\omega_{\gamma(s)}(\gamma'(s))\|\le C\|\gamma'(s)\|$ of
Remark~\ref{def:connection}.
\end{proof}

\section{Curvature and Holonomy}
\label{sec:holonomy}

In the discrete setting of the companion paper, the cocycle identity
forces all cycle products $K_{i_1 i_2}K_{i_2 i_3}\cdots K_{i_k i_1}$
to equal the identity: there are no non-contractible loops and hence no
holonomy. In the continuous setting this is no longer automatic. A
regular path intertwining cocycle defines a connection on
$\mathcal{H}\times M$ (Section~\ref{sec:connection}); the connection
may fail to be flat, or may be flat with nontrivial monodromy, and in
either case a closed loop $\gamma$ in $M$ carries a holonomy operator
$K_\gamma\in\mathcal{B}(\mathcal{H})$ that need not be the identity.
We define these objects, constrain them, and characterize flatness in
terms of parallel transport.

\begin{definition}[Holonomy]
\label{def:holonomy}
Let $x_0\in M$. For a piecewise smooth closed loop
$\gamma\colon[0,1]\to M$ with $\gamma(0)=\gamma(1)=x_0$, the
\emph{holonomy operator} of a path intertwining cocycle along
$\gamma$ is the transport operator $K_\gamma\in\mathcal{B}(\mathcal{H})$
of Definition~\ref{def:path-cocycle}. By \emph{(P1)}, $K_\gamma$ is
boundedly invertible, with $K_\gamma^{-1}=K_{\bar\gamma}$ where
$\bar\gamma$ denotes $\gamma$ traversed in reverse. The
\emph{holonomy group} at $x_0$ is
$\mathrm{Hol}(x_0) := \{K_\gamma : \gamma \text{ closed at }x_0\}
\subset\mathcal{B}(\mathcal{H})^\times$.
\end{definition}

\begin{proposition}[Holonomy is spectrally block-diagonal]
\label{prop:commutant}
Let $\gamma$ be a piecewise smooth closed loop at $x_0$. Then
$\lambda(\gamma)=1$ and
\begin{equation}
\label{eq:holonomy-commutes}
K_\gamma\,\mathcal{S}_{x_0}(t)=\mathcal{S}_{x_0}(t)\,K_\gamma
\qquad(t\ge0).
\end{equation}
Consequently $K_\gamma$ commutes with every spectral projection of
$A_{x_0}$ and restricts to a bijection of each eigenspace
$E_\alpha(x_0)=\ker(A_{x_0}-\alpha I)$. In particular
$\mathrm{Hol}(x_0)$ is contained in the group of boundedly invertible
elements of the commutant $\{A_{x_0}\}'$.
\end{proposition}

\begin{proof}
$\lambda(\gamma)=1$ is Lemma~\ref{lem:no-scalar-monodromy}, so
\emph{(P3)} reads \eqref{eq:holonomy-commutes}. Commutation with the
semigroup implies commutation with the resolvent
$(A_{x_0}+\mu)^{-1}=\int_0^\infty e^{-\mu t}\,\mathcal{S}_{x_0}(t)\,dt$
for $\mu>0$, hence with every spectral projection of $A_{x_0}$ by the
spectral theorem. Each eigenspace is therefore invariant under
$K_\gamma$ and under $K_\gamma^{-1}=K_{\bar\gamma}$, so the
restriction to it is bijective.
\end{proof}

\begin{assumption}[Regularity of the connection form]\label{ass:Dprime}
The connection form $\omega$ of Assumption~\ref{ass:D} is
continuous on $M$ and admits continuous directional derivatives
$\partial_u\omega_x(v)$ for all $u,v\in T_xM$.
\end{assumption}

\begin{proposition}[Existence of the curvature limit]
\label{prop:curvature-exists}
Under Assumptions~\ref{ass:D} and~\ref{ass:Dprime}, for every $x\in M$
and every $u,v\in T_xM$, the limit
\begin{equation}
\label{eq:curvature-limit}
\lim_{\varepsilon\to 0}\frac{1}{\varepsilon^2}
\left(K_{\partial\Pi_\varepsilon} - I\right)h
\end{equation}
exists in $\mathcal{H}$ for every $h\in\mathcal{H}$, uniformly in
$h$ on bounded sets, and equals $-F_x(u,v)\,h$, where
$F_x(u,v)\colon\mathcal{H}\to\mathcal{H}$ is a bounded operator,
skew-symmetric in $u,v$.
\end{proposition}

\begin{proof}
Let $\partial\Pi_\varepsilon
=\bar\gamma_2^{\,\varepsilon v}\cdot\bar\gamma_1^{\,\varepsilon u}
\cdot\gamma_2^{\,\varepsilon v}\cdot\gamma_1^{\,\varepsilon u}$
be the boundary of the coordinate parallelogram spanned by
$\varepsilon u,\varepsilon v$ at $x$ (sides traversed
$u$-direction, $v$-direction, then back). Since $\omega$ is
norm-continuous with continuous directional derivatives
(Assumption~\ref{ass:Dprime}), each side admits the norm-convergent
Dyson expansion of its transport to second order: for the side from
$y$ in direction $\varepsilon w$,
\begin{equation}
\label{eq:arc-expansion}
K = I - \varepsilon\,\omega_y(w)
- \tfrac{\varepsilon^2}{2}\,\partial_w\omega_y(w)
+ \tfrac{\varepsilon^2}{2}\,\omega_y(w)^2
+ o(\varepsilon^2),
\end{equation}
in operator norm, where the middle terms come from expanding
$\omega$ along the side and from the second Dyson term. Multiplying
the four expansions \eqref{eq:arc-expansion} for the four sides
(with base points $x$, $x+\varepsilon u$, $x+\varepsilon v$, $x$ up
to $o(\varepsilon)$, expanded via
$\omega_{x+\delta}=\omega_x+\partial_\delta\omega_x+o(\|\delta\|)$)
and collecting terms, all first-order terms cancel, the
$\tfrac{\varepsilon^2}{2}$-terms cancel in opposite-side pairs, and
the surviving second-order contribution is
\begin{equation}
\label{eq:parallelogram-expansion}
K_{\partial\Pi_\varepsilon} - I
= -\varepsilon^2\,\bigl(\partial_u\omega_x(v)-\partial_v\omega_x(u)
+[\omega_x(u),\omega_x(v)]\bigr)
+ o(\varepsilon^2)
\end{equation}
in operator norm; this is the standard curvature computation for
connections with values in a Banach algebra, carried out here with
the sign convention of \eqref{eq:pt-ode}. Dividing by
$\varepsilon^2$ and letting $\varepsilon\to0$ yields the limit
$-F_x(u,v)h$ with
$F_x(u,v)=d\omega_x(u,v)+[\omega_x(u),\omega_x(v)]$, where
$d\omega_x(u,v)=\partial_u\omega_x(v)-\partial_v\omega_x(u)$;
boundedness is clear from Assumption~\ref{ass:Dprime}, and
skew-symmetry follows from orientation reversal of
$\partial\Pi_\varepsilon$ together with \emph{(P1)}.
\end{proof}

\begin{definition}[Curvature]
\label{def:curvature}
Under Assumptions~\ref{ass:D} and~\ref{ass:Dprime}, the
\emph{curvature} of the operator connection is the
$\mathcal{B}(\mathcal{H})$-valued $2$-form $F$ on $M$ whose value at
$x\in M$ on tangent vectors $u,v\in T_xM$ is
\begin{equation}
\label{eq:curvature}
F_x(u,v) \;=\; d\omega_x(u,v) + [\omega_x(u),\,\omega_x(v)],
\end{equation}
where $\omega$ is the connection $1$-form of
Assumption~\ref{ass:D} and $[\cdot,\cdot]$ is the commutator
in $\mathcal{B}(\mathcal{H})$; equivalently,
$F_x(u,v)h=-\lim_{\varepsilon\to0}\varepsilon^{-2}
(K_{\partial\Pi_\varepsilon}-I)h$ by
Proposition~\ref{prop:curvature-exists}.
\end{definition}

\begin{theorem}[Flatness characterization]
\label{thm:flatness}
Under Assumptions~\ref{ass:D} and~\ref{ass:Dprime}, the following
conditions are equivalent:

\noindent\emph{(i)} $F = 0$ (the connection has zero curvature).

\noindent\emph{(ii)} $K_\gamma = I$ for every contractible closed loop
$\gamma$ in $M$.

\noindent\emph{(iii)} The transport $K_\gamma$ depends only on
the homotopy class of $\gamma$ relative to its endpoints.

\noindent
In particular, if $M$ is simply connected, then flatness is equivalent
to the transport being determined by its endpoints alone, i.e.\ to the
path cocycle descending to a two-point cocycle in the sense of
Definition~\ref{def:cocycle}.
\end{theorem}

\begin{proof}
The argument adapts the finite-dimensional flatness criterion for
connections on principal bundles
\cite[Ch.~II, Thm.~9.1]{kobayashi-nomizu} to the
$\mathcal{B}(\mathcal{H})$-valued setting; all steps below use only
the ODE \eqref{eq:pt-ode} with its norm-continuous coefficient and
the norm estimates of Remark~\ref{def:connection}.

\emph{(i) $\Rightarrow$ (ii).} Suppose $F=0$. Let
$H\colon[0,1]^2\to M$ be a smooth homotopy of loops at $x_0$,
$H(s,\cdot)=\gamma_s$, with $\gamma_0$ the constant loop, and write
$P_s(t)$ for the transport along $\gamma_s|_{[0,t]}$. The homotopy
variation formula for connections with values in a Banach algebra
--- obtained by differentiating \eqref{eq:pt-ode} in $s$ and
integrating by parts in $t$, the endpoint terms vanishing because
the loops are based at $x_0$ --- reads
\begin{equation*}
\partial_s P_s(1)
= -\,P_s(1)\int_0^1 P_s(t)^{-1}\,
F_{H(s,t)}\!\bigl(\partial_sH,\partial_tH\bigr)\,P_s(t)\,dt.
\end{equation*}
Since $F\equiv0$, $\partial_sP_s(1)=0$, hence
$K_{\gamma}=P_1(1)=P_0(1)=I$.

\emph{(ii) $\Rightarrow$ (iii).} If $\gamma_1,\gamma_2$ are
homotopic paths from $y$ to $x$, then
$\gamma_1\cdot\bar\gamma_2$ is contractible at $y$, so
$K_{\gamma_1\cdot\bar\gamma_2}=I$ by (ii), hence
$K_{\gamma_1}K_{\gamma_2}^{-1}=I$, i.e.\ $K_{\gamma_1}=K_{\gamma_2}$.

\emph{(iii) $\Rightarrow$ (i).} Homotopy invariance implies
$K_{\partial\Pi_\varepsilon}=I$ for every contractible
parallelogram $\Pi_\varepsilon$. By
Proposition~\ref{prop:curvature-exists},
$F_x(u,v)h = -\lim_{\varepsilon\to 0}\varepsilon^{-2}(K_{\partial\Pi_\varepsilon}-I)h=0$.
\end{proof}

The flatness characterization, the commutant constraint of
Proposition~\ref{prop:commutant}, and the scalar rigidity of
Lemma~\ref{lem:no-scalar-monodromy} combine into a complete
description of the topological sector.

\begin{remark}[Networks and cocycles on Hilbert bundles]
\label{rem:bundle-extension}
The notions of operator network and path intertwining cocycle extend
verbatim fiberwise to a Hilbert bundle $E\to M$ with fiber
$\mathcal{H}$: the generators act in the fibers $E_x$ with a dense
domain subbundle, the transport operators are fiber maps
$K_\gamma\colon E_{\gamma(0)}\to E_{\gamma(1)}$ satisfying
\emph{(P1)}--\emph{(P4)} fiberwise, and regularity is required in
local trivializations. All results of
Sections~\ref{sec:connection}--\ref{sec:holonomy} hold in this
setting with identical proofs; the trivial bundle treated so far is
merely notationally lightest. This extension is used in the
realization part of the following theorem.
\end{remark}

\begin{theorem}[Classification of flat holonomy]
\label{thm:classification}
Let $\{K_\gamma,\lambda\}$ be a flat path intertwining cocycle on a
continuous operator network over the connected manifold $M$, and fix
$x_0\in M$.

\noindent\emph{(i)} The assignment $\gamma\mapsto K_\gamma$ descends
to a group homomorphism
\begin{equation}
\label{eq:holonomy-rep}
\rho\colon\pi_1(M,x_0)\;\longrightarrow\;
\mathcal{G}\bigl(\{A_{x_0}\}'\bigr),
\end{equation}
where $\mathcal{G}(\{A_{x_0}\}')$ denotes the group of boundedly
invertible elements of the commutant of $A_{x_0}$. If the cocycle is
unitary, the image lies in the unitary group
$\mathcal{U}(\{A_{x_0}\}')$.

\noindent\emph{(ii)} Conversely, let $A_0$ be a positive self-adjoint
operator with compact resolvent and let
$\rho\colon\pi_1(M,x_0)\to\mathcal{U}(\{A_0\}')$ be any homomorphism.
Then $\rho$ is the holonomy representation of a flat unitary path
intertwining cocycle with $\lambda\equiv1$, canonically realized on
the flat Hilbert bundle $E_\rho$ with fiberwise generator $A_0$.
If $\mathcal{H}$ is infinite-dimensional, $E_\rho$ is trivializable,
and whenever the trivialization $T=\{T_x\}$ can be chosen $C^1$ and
domain-preserving, the construction transfers to a regular flat
cocycle on the trivial bundle $\mathcal{H}\times M$ over the
isospectral network $A_x=T_xA_0T_x^{-1}$. If, in addition, the flat
principal $\mathcal{U}(\{A_0\}')$-bundle associated with $\rho$ is
trivializable --- as is always the case for $M=S^1$ --- the
realization can be chosen with the constant network
$A_x\equiv A_0$.
\end{theorem}

\begin{proof}
\emph{(i).} By flatness and Theorem~\ref{thm:flatness}\emph{(iii)},
$K_\gamma$ depends only on the homotopy class
$[\gamma]\in\pi_1(M,x_0)$; concatenation \emph{(P2)} makes the
assignment multiplicative and \emph{(P1)} sends inverse classes to
inverse operators, so $\rho$ is a homomorphism into
$\mathcal{B}(\mathcal{H})^\times$. Its image lies in the commutant by
Proposition~\ref{prop:commutant}.

\emph{(ii).} This is the monodromy construction for flat bundles.
Let $p\colon\widetilde M\to M$ be the universal cover, on which
$\pi_1(M,x_0)$ acts by deck transformations, and form the flat
Hilbert bundle
$E_\rho:=(\widetilde M\times\mathcal{H})/\pi_1(M,x_0)$, where the
group acts by $g\cdot(\tilde m,h)=(g\tilde m,\rho(g)h)$. Parallel
transport along a path in $M$ is induced by the constant transport
$(\tilde m,h)\mapsto(\tilde m',h)$ on $\widetilde M\times\mathcal{H}$,
and its monodromy is $\rho$. Because $\rho$ takes values in
$\{A_0\}'$, the constant operator family $A_0$ on
$\widetilde M\times\mathcal{H}$ descends to a well-defined fiberwise
generator on $E_\rho$, equal to $A_0$ in every flat local
trivialization, and the transport intertwines the associated
semigroups with $\lambda\equiv1$; in flat charts the transport is
locally constant, so the cocycle is regular with $\omega\equiv0$
there and $F=0$. This proves the canonical realization. If
$\mathcal{H}$ is infinite-dimensional, the unitary group of
$\mathcal{H}$ is contractible by Kuiper's
theorem~\cite{kuiper1965}, so $E_\rho$ admits a global
trivialization $T=\{T_x\}$; whenever $T$ can be chosen $C^1$ and
domain-preserving, pushing the construction forward yields a regular
flat cocycle on $\mathcal{H}\times M$ intertwining the isospectral
family $A_x=T_xA_0T_x^{-1}$. Note that $T_x$ need not commute with
$A_0$, so the pushed-forward network is in general not constant; it
is constant precisely when the trivialization can be chosen with
values in $\mathcal{U}(\{A_0\}')$, i.e.\ when the associated flat
principal $\mathcal{U}(\{A_0\}')$-bundle is trivializable. For $M=S^1$ this is
always possible, by an explicit construction: given
$V:=\rho(1)\in\mathcal{U}(\{A_0\}')$, the commutant is a von Neumann
algebra, so $V=e^{B}$ for a skew-adjoint $B\in\{A_0\}'$ with
$\|B\|\le\pi$; the constant-coefficient form
$\omega:=-(B/2\pi)\,d\theta$ then defines a regular flat cocycle on
the constant network $A_x\equiv A_0$ with holonomy
$e^{B}=V$, since $\omega$ takes values in the commutant and
therefore intertwines the constant semigroup with
$\lambda\equiv1$.
\end{proof}

\begin{remark}[The topological dichotomy]
\label{rem:dichotomy}
Together with Lemma~\ref{lem:no-scalar-monodromy},
Theorem~\ref{thm:classification} gives a complete picture of the
topological sector of the theory: the scaling layer never sees
$\pi_1(M)$, while the operator layer sees it exactly through unitary
representations in the commutant of the generator --- that is, through
the spectral multiplicity data of $A_{x_0}$. For a generator with
simple spectrum the unitary commutant is the group of diagonal phase
operators, and holonomy reduces to one unitary phase per level,
decomposing into the Berry phase and the sector twist
(Section~\ref{sec:berry}); degenerate eigenvalues admit non-Abelian
blocks in the sense of Wilczek--Zee~\cite{wilczek-zee1984}.
\end{remark}

\begin{remark}[Discrete case is always flat]
\label{rem:discrete-flat}
When $M=I$ is finite with the discrete topology, every continuous path
is constant, so path transport is trivially the identity and holonomy
is vacuous. The combinatorial cycles of the companion paper --- finite
sequences $i_0\to i_1\to\cdots\to i_k=i_0$ --- are compositions of
two-point transports, which the cocycle identity \eqref{eq:cocycle-C}
collapses to $K_{i_0,i_0}=I$. This is the precise sense in which the
companion paper lives in the flat sector of the present theory.
\end{remark}

\section{Relation to Adiabatic Theory}
\label{sec:berry}

The connection $\omega$ defined in Section~\ref{sec:connection} acts on
the full Hilbert bundle $\mathcal{H}\times M$. A classical object in
mathematical physics --- the Berry connection of adiabatic quantum
mechanics \cite{berry1984} --- is a $U(1)$-connection on the
eigenspace line bundle $E_\alpha\to M$ associated to a simple
eigenvalue $\alpha(x)$ of the Hamiltonian $A_x$. We show that the
restriction of the transport to each renormalized spectral sector
decomposes canonically as adiabatic (Berry--Wilczek--Zee) transport twisted by
an independent commutant-valued \emph{sector potential} --- the
spectral-sector block of the connection form. Adiabatic theory is
thus a
finite-rank sector of the present framework, and it captures the
transport exactly on the class of \emph{adiabatically normalized}
cocycles, for which the sector potentials vanish. This is not a
technicality: by the classification of
Section~\ref{sec:holonomy}, arbitrary commutant-valued twists do
occur, so the sector potential is genuine freedom, not a removable
gauge artifact.

\begin{theorem}[Sector decomposition of the transport]
\label{prop:berry}
Let $\{K_\gamma,\lambda\}$ be a regular path intertwining cocycle
(Assumption~\ref{ass:D}) on a network as in
Section~\ref{sec:networks}, let $\tau$ be its gauge function
(Lemma~\ref{lem:no-scalar-monodromy}), and let
$B_x:=\tau(x)A_x$ be the renormalized generators, so that
$\Sigma:=\sigma(B_x)$ is independent of $x$
(Corollary~\ref{cor:spectral-rigidity}). Fix $\mu\in\Sigma$ and set
$E_\mu(x):=\ker(B_x-\mu I)$, with orthogonal projection $P_\mu(x)$.
Fix a smooth path $\gamma$ and assume that
$s\mapsto P_\mu(\gamma(s))$ has constant finite rank $n$ and is
$C^1$ in operator norm; then a $C^1$ orthonormal frame
$\{e_j(s)\}_{j=1}^n$ of $E_\mu(\gamma(s))$ exists along $\gamma$
(a global frame over $M$ need not exist, see
Remark~\ref{rem:mobius}).

\noindent\emph{(a)} $K_{\gamma|_{[0,s]}}$ maps
$E_\mu(\gamma(0))$ onto $E_\mu(\gamma(s))$.

\noindent\emph{(b)} The frame matrix $c(s)$, defined by
$K_{\gamma|_{[0,s]}}\,e_j(0)=\sum_{i}c_{ij}(s)\,e_i(s)$, satisfies
\begin{equation}
\label{eq:berry-recovery}
\dot c(s)
= -\bigl(\mathcal{A}^{\mu}(s)+\Omega^{\mu}(s)\bigr)\,c(s),
\qquad c(0)=I_n,
\end{equation}
where
$\mathcal{A}^{\mu}_{ij}(s)
=\langle e_i(s),\tfrac{d}{ds}e_j(s)\rangle$
is the Berry--Wilczek--Zee connection of the frame and
\begin{equation}
\label{eq:sector-potential}
\Omega^{\mu}_{ij}(s)
=\bigl\langle e_i(s),\,
\omega_{\gamma(s)}\!\left(\dot\gamma(s)\right)e_j(s)\bigr\rangle
\end{equation}
is the \emph{sector potential} --- the $E_\mu$-block of the
connection form. For unitary cocycles both matrices are
skew-Hermitian and $c(s)$ is unitary.

\noindent\emph{(c)} The transport
$K_{\gamma|_{[0,s]}}\big|_{E_\mu}$ coincides with the adiabatic
(Berry--Wilczek--Zee) transport \emph{for every} $s\in[0,1]$ if and
only if $\Omega^{\mu}(s)=0$ for all $s$; for the endpoint operator
$K_\gamma|_{E_\mu}$ alone, vanishing of the sector potential is
sufficient. We call the cocycle \emph{adiabatically normalized} if
$P_\mu(x)\,\omega_x(v)\,P_\mu(x)=0$ for every $\mu\in\Sigma$, every
$x\in M$ and $v\in T_xM$; for such cocycles the restriction of the
transport to every spectral sector is exactly the adiabatic
transport.
\end{theorem}

\begin{proof}
\emph{(a).} By the generator transport of
Remark~\ref{rem:two-point} (with the domain identity of
Proposition~\ref{prop:generator}), writing
$\gamma_s:=\gamma|_{[0,s]}$ and multiplying
$K_{\gamma_s}A_{\gamma(0)}
=\lambda(\gamma_s)A_{\gamma(s)}K_{\gamma_s}$ by $\tau(\gamma(0))$
and using $\lambda(\gamma_s)=\tau(\gamma(s))/\tau(\gamma(0))$
(Lemma~\ref{lem:no-scalar-monodromy}),
\begin{equation*}
K_{\gamma_s}\,B_{\gamma(0)}
= B_{\gamma(s)}\,K_{\gamma_s}
\qquad\text{on }\mathcal{D}(B_{\gamma(0)}).
\end{equation*}
Hence for $\psi\in E_\mu(\gamma(0))$:
$B_{\gamma(s)}K_{\gamma_s}\psi=K_{\gamma_s}B_{\gamma(0)}\psi
=\mu\,K_{\gamma_s}\psi$, i.e.\
$K_{\gamma_s}\psi\in E_\mu(\gamma(s))$ --- no choice of eigenvalue
branch is involved. Applying the same argument to the reversed path
and $K_{\gamma_s}^{-1}$ gives surjectivity.

\emph{(b).} By part \emph{(a)} the expansion
$K_{\gamma_s}e_j(0)=\sum_i c_{ij}(s)e_i(s)$ is well defined, and the
frame is differentiable by the $C^1$ hypothesis on $P_\mu$.
Differentiating it in $s$ and using \eqref{eq:pt-ode},
\begin{equation*}
-\,\omega_{\gamma(s)}(\dot\gamma(s))\,K_{\gamma_s}e_j(0)
=\sum_i \dot c_{ij}(s)\,e_i(s)+\sum_i c_{ij}(s)\,\dot e_i(s).
\end{equation*}
Substituting the expansion on the left and taking the inner product
with $e_l(s)$ annihilates the components of
$\omega\,e_k$ orthogonal to $E_\mu(\gamma(s))$ and yields
\begin{equation*}
-\sum_k \Omega^\mu_{lk}(s)\,c_{kj}(s)
=\dot c_{lj}(s)+\sum_i \mathcal{A}^\mu_{li}(s)\,c_{ij}(s),
\end{equation*}
which is \eqref{eq:berry-recovery}. For unitary cocycles,
skew-Hermiticity of $\mathcal{A}^\mu$ follows from
differentiating $\langle e_i,e_j\rangle=\delta_{ij}$, and that of
$\Omega^\mu$ from unitarity of the transport; hence $c(s)$
solves a linear ODE with skew-Hermitian coefficient and is unitary.

\emph{(c).} If $\Omega^\mu\equiv0$ along $\gamma$, then
\eqref{eq:berry-recovery} is the adiabatic transport equation and
the two transports coincide for every $s$ (in particular at the
endpoint). Conversely, if $c(s)=c^{\mathrm{ad}}(s)$ for all
$s\in[0,1]$, differentiating and subtracting the two ODEs gives
$\Omega^\mu(s)\,c(s)=0$ with $c(s)$ invertible, hence
$\Omega^\mu(s)=0$ for all $s$. For the endpoint operator alone the
converse may fail: distinct sector potentials can produce the same
path-ordered exponential at $s=1$.
\end{proof}

\begin{remark}[Consistency with the classification]
\label{rem:sector-consistency}
The sector potential is exactly the freedom catalogued by
Theorem~\ref{thm:classification}. For instance, on $M=S^1$ with the
constant network $A_x\equiv A_0$, the cocycle
$K_\gamma=e^{\,i\int_\gamma\beta}\,I$, for any real $1$-form $\beta$
(automatically closed in one dimension), is regular and flat with
$\omega=-i\beta\,I$, has vanishing Berry--Wilczek--Zee data, and
transports every eigenline with the phase
$e^{\,i\int_\gamma\beta}$: this is precisely the $S^1$-realization
construction in the proof of
Theorem~\ref{thm:classification}\emph{(ii)}. (On a general base the
same formula defines a regular cocycle for any $\beta$, flat
precisely when $d\beta=0$.) Adiabatic
normalization removes exactly this commutant-valued gauge freedom;
without it, eigenspace data alone cannot determine the transport.
\end{remark}

\begin{corollary}[Holonomy assembles the sector monodromies]
\label{cor:assembly}
Let $\{K_\gamma\}$ be a regular path intertwining cocycle and
$\gamma$ a piecewise smooth closed loop at $x_0$.

\noindent\emph{(i)} $K_\gamma$ preserves every spectral sector
$E_\mu(x_0)$ (equivalently, every eigenspace of $A_{x_0}$), and,
whenever $P_\mu$ is $C^1$ of constant finite rank along $\gamma$,
its restriction to $E_\mu(x_0)$ is the holonomy of the sector
connection $\mathcal{A}^\mu+\Omega^\mu$ of
Theorem~\ref{prop:berry}\emph{(b)}.

\noindent\emph{(ii)} Consequently
$K_\gamma=\bigoplus_\mu K_\gamma|_{E_\mu(x_0)}$: the operator
holonomy is the assembly of the sector monodromies of all spectral
levels. It is the assembly of the adiabatic
(Berry--Wilczek--Zee~\cite{wilczek-zee1984}) monodromies whenever
the cocycle is adiabatically normalized.

\noindent\emph{(iii)} The holonomy is not determined by the Berry
connection one-forms: Example~\ref{ex:holonomy} exhibits an
adiabatically normalized network in which the Berry one-form of
every level vanishes identically while $K_\gamma$ is a nontrivial
unitary. The monodromy is carried by the orientation
($\mathbb{Z}_2$) classes of the real eigenline bundles, which are
invisible to the local connection data.
\end{corollary}

\begin{proof}
Sector invariance is Proposition~\ref{prop:commutant} (note that at
the base point $E_\mu(x_0)=\ker(A_{x_0}-\mu/\tau(x_0))$, so the
sectors are the eigenspaces of $A_{x_0}$); the identification of the
restricted action is Theorem~\ref{prop:berry}. For \emph{(ii)}, the
eigenspaces of $A_{x_0}$ span $\mathcal{H}$ (compact resolvent), and
the bounded operator $K_\gamma$ is determined by its restrictions to
them; the normalized case is Theorem~\ref{prop:berry}\emph{(c)}.
\emph{(iii)} is verified in Example~\ref{ex:holonomy} and
Remark~\ref{rem:mobius}.
\end{proof}

\section{Example: A Rotating Network with Parity Holonomy}
\label{sec:example}

We exhibit an explicit regular path intertwining cocycle with
nontrivial holonomy, confirming that the monodromy phenomenon of
Section~\ref{sec:holonomy} is not merely formal. The base manifold is
$M = S^1=\mathbb{R}/2\pi\mathbb{Z}$, and the Hilbert space is
$\mathcal{H}=\ell^2(\mathbb{N})$ with orthonormal basis
$\{e_n\}_{n\ge1}$.

\begin{example}[Rotating network with parity holonomy]
\label{ex:holonomy}
Let $A_0e_n=n\,e_n$ on its natural domain
$\mathcal{D}(A_0)=\{h:\sum n^2|h_n|^2<\infty\}$ (positive,
compact resolvent, simple spectrum), and let $P$ denote the parity
operator $Pe_n=(-1)^ne_n$. Define the bounded real skew-adjoint
operator $J$ by
\begin{equation}
\label{eq:oscillator}
Je_{4k+1}=\tfrac12\,e_{4k+3},\qquad
Je_{4k+3}=-\tfrac12\,e_{4k+1}\quad(k\ge0),
\qquad Je_n=0\ \text{for even }n,
\end{equation}
so that $\|J\|=\tfrac12$, $e^{2\pi J}=P$, and $[J,A_0]\neq0$. Set
$U_x:=e^{xJ}$ (a norm-continuous one-parameter unitary group) and
\begin{equation}
\label{eq:rotating-family}
A_x := U_{\tilde x}\,A_0\,U_{\tilde x}^{-1},
\qquad \tilde x\in\mathbb{R}\ \text{a lift of }x\in S^1.
\end{equation}
The family is well defined on $S^1$, since replacing $\tilde x$ by
$\tilde x+2\pi$ conjugates $A_0$ additionally by $e^{2\pi J}=P$,
which commutes with $A_0$. For a piecewise smooth path $\gamma$ in
$S^1$ let $\tilde\gamma$ be any lift and set
\begin{equation}
\label{eq:rotating-cocycle}
K_\gamma := U_{\tilde\gamma(1)-\tilde\gamma(0)},
\qquad
\lambda\equiv1;
\end{equation}
the displacement $\tilde\gamma(1)-\tilde\gamma(0)$ does not depend
on the choice of lift. Then $\{K_\gamma,\lambda\equiv1\}$ is a
regular path intertwining cocycle on the network
$\{A_x\}_{x\in S^1}$, with connection form
$\omega_x(v)=-\tilde v\,J$; it is flat; it is adiabatically
normalized; its holonomy along the generating loop $\gamma_1$ of
$\pi_1(S^1)$ is
\begin{equation}
\label{eq:holonomy-result}
K_{\gamma_1}=e^{2\pi J}=P\;\neq\;I;
\end{equation}
and the Berry connection one-form of every eigenline bundle
vanishes identically.
\end{example}

\begin{proof}
\emph{Path cocycle axioms and regularity.} Displacements of lifts
are additive under concatenation, change sign under reversal, vanish
for constant paths, and are invariant under orientation-preserving
reparametrization; since $x\mapsto U_x$ is a one-parameter group,
\emph{(P1)}--\emph{(P2)} follow, and each $K_\gamma$ is unitary. For
$P(s)=U_{\tilde\gamma(s)-\tilde\gamma(0)}$ we have, in operator
norm,
$P'(s)=\tilde\gamma'(s)\,J\,P(s)$,
because $J$ is bounded; hence Assumption~\ref{ass:D} holds with the
connection form $\omega_x(v)=-\tilde v\,J$, which is constant in $x$
and in particular satisfies Assumption~\ref{ass:Dprime}. Since
$\omega$ is constant and commutes with itself,
$F=d\omega+[\omega,\omega]=0$: the cocycle is flat, consistent with
the fact that the transport depends only on the displacement of the
lift, so contractible loops have transport $I$.

\emph{The network.} For $h\in\mathcal{D}(A_0)$,
\begin{equation*}
\|A_0Jh\|^2
=\sum_{k\ge0}\Bigl[\tfrac{(4k+3)^2}{4}\,|h_{4k+1}|^2
+\tfrac{(4k+1)^2}{4}\,|h_{4k+3}|^2\Bigr]
\le\tfrac94\,\|A_0h\|^2,
\end{equation*}
so $J$ is bounded in the graph norm of $A_0$ and
$e^{xJ}\mathcal{D}(A_0)=\mathcal{D}(A_0)$: the common domain
$D=\mathcal{D}(A_0)$ satisfies Assumption~\ref{ass:R}. Moreover
$(A_x-z)^{-1}=U_x(A_0-z)^{-1}U_x^{-1}$, so $x\mapsto A_x$ is even
norm-resolvent continuous, $x\mapsto U_x$ being norm-continuous;
each $A_x$ is positive self-adjoint with compact resolvent and
$\sigma(A_x)=\{n\}_{n\ge1}$. Since $\tau\equiv\mathrm{const}$ here,
the spectral sectors of Theorem~\ref{prop:berry} are the eigenlines
of $A_x$.

\emph{Intertwining \emph{(P3)} with $\lambda\equiv1$.} With lifts
$\tilde y=\tilde\gamma(0)$, $\tilde x=\tilde\gamma(1)$,
\begin{equation*}
K_\gamma\,A_{\gamma(0)}\,K_\gamma^{-1}
=U_{\tilde x-\tilde y}\,U_{\tilde y}\,A_0\,U_{\tilde y}^{-1}\,
U_{\tilde y-\tilde x}
=U_{\tilde x}\,A_0\,U_{\tilde x}^{-1}=A_{\gamma(1)},
\end{equation*}
and hence
$K_\gamma\mathcal{S}_{\gamma(0)}(t)=\mathcal{S}_{\gamma(1)}(t)K_\gamma$.

\emph{Holonomy \eqref{eq:holonomy-result}.} The generating loop
lifts to a path of displacement $2\pi$, so
$K_{\gamma_1}=e^{2\pi J}=P$, and a loop winding $k$ times has
holonomy $P^k$. The holonomy representation of
Theorem~\ref{thm:classification} is
$\rho\colon\pi_1(S^1)\cong\mathbb{Z}\to\mathcal{U}(\{A_{x_0}\}')$,
$\rho(k)=P^k$, and $P$ commutes with $A_{x_0}$ in accordance with
Proposition~\ref{prop:commutant}.

\emph{Normalization and vanishing Berry forms.} The eigenvector
field of the $n$-th line along a path with lift $\tilde x(s)$ is
$\phi_n(s)=U_{\tilde x(s)}e_n$, real in the real structure of
$\ell^2$ preserved by the real operator $J$. The sector potential
and the Berry one-form are
\begin{equation*}
\Omega^{n}(\partial_s)
=-\tilde x'(s)\,\langle U e_n,\,J\,U e_n\rangle
=-\tilde x'(s)\,\langle e_n,\,Je_n\rangle=0,
\qquad
\mathcal{A}^{n}(\partial_s)
=\tilde x'(s)\,\langle e_n,\,Je_n\rangle=0,
\end{equation*}
since $J$ is real and skew-symmetric while $e_n$ is real. Hence the
cocycle is adiabatically normalized and all Berry one-forms vanish.

\emph{Holonomy on eigenlines.} $Pe_n=(-1)^ne_n$, so $K_{\gamma_1}$
acts on the $n$-th eigenline as the sign $(-1)^n$ --- equal to $-1$
on infinitely many levels --- whence $K_{\gamma_1}\neq I$ while
every $\oint\mathcal{A}^{n}=0$ and every $\Omega^{n}=0$.
\end{proof}

\begin{remark}[M\"{o}bius eigenline bundles]
\label{rem:mobius}
For odd $n$ the eigenvector field does not close up:
$\phi_n(2\pi)=Pe_n=-e_n=-\phi_n(0)$. The real eigenline bundle over
$S^1$ is then non-orientable --- a M\"{o}bius band --- and no
globally periodic smooth real unit eigenvector field exists over
$S^1$; the pathwise frame required by Theorem~\ref{prop:berry} does
exist along $[0,1]$, but returns with the sign $-1$, which is
precisely how the sign escapes the Berry one-form. This
$\mathbb{Z}_2$ monodromy is the infinite-dimensional counterpart of
the Longuet--Higgins sign change of real eigenfunctions transported
around a conical
intersection~\cite{herzberg-longuet-higgins1963,berry1984}.
\end{remark}

\begin{remark}[Physical illustration: the rotating anisotropic
oscillator]
\label{rem:oscillator}
The same mechanism is realized by a Schr\"{o}dinger family: let
$A_0=-\Delta+a^2u_1^2+b^2u_2^2$ on $L^2(\mathbb{R}^2)$ with
$a/b\notin\mathbb{Q}$, let $U_\theta$ be the rotation group of the
plane, and set $A_x=U_{\tilde x/2}A_0U_{\tilde x/2}^{-1}$; since
$U_\pi$ is the parity operator, which fixes the even potential, the
family descends to $S^1$, the displacement cocycle has holonomy
equal to parity, the eigenfunctions (Hermite products) are real, and
the odd eigenline bundles are M\"{o}bius bands exactly as above.
However, the transport generator --- the angular momentum operator
--- is unbounded: $s\mapsto K_{\gamma|_{[0,s]}}h$ is differentiable
only for $h$ in the dense core $\mathcal{D}(G)$, so this cocycle
satisfies \emph{(P1)}--\emph{(P4)} but falls outside the bounded
regularity class of Assumption~\ref{ass:D}. A theory of unbounded
operator connections on invariant cores would be required to bring
it within the scope of the theorems above; we leave this extension
for future work.
\end{remark}

\begin{remark}[Variants and a non-example]
\label{rem:variants}
\emph{Nontrivial time-scaling.} In the example $\lambda\equiv 1$ and
$\tau\equiv\mathrm{const}$. To
combine nontrivial holonomy with a nontrivial gauge function, replace
$A_x$ by $B_x:=\varrho(x)A_x$ for a smooth
$\varrho\colon S^1\to\mathbb{R}_{>0}$. The same transport operators
satisfy
$K_\gamma B_{\gamma(0)}
=\bigl(\varrho(\gamma(0))/\varrho(\gamma(1))\bigr)\,
B_{\gamma(1)}K_\gamma$,
i.e.\ $\lambda(\gamma)=\tau(\gamma(1))/\tau(\gamma(0))$ with
$\tau=1/\varrho$, in accordance with
Lemma~\ref{lem:no-scalar-monodromy}: the scaling layer remains
single-valued on $S^1$, while the holonomy computation is unchanged.
Both phenomena --- nontrivial gauge and nontrivial holonomy --- thus
coexist.

\emph{Shift families.}
A natural first attempt at holonomy is the shifted family
$A_x=-d^2/du^2+V_0(u-f(x))$ on $L^2(\mathbb{R})$, with a confining
potential $V_0$, translations as transport, and
$f(2\pi)-f(0)=\delta\neq0$. This is a perfectly good cocycle over the
base $\mathbb{R}$ --- where every loop is contractible and no
monodromy is available --- but not over $S^1$: since $f$ does not
descend to the circle, neither does the family ($A_{x+2\pi}\neq A_x$
unless $V_0$ is $\delta$-periodic, which is excluded for confining
$V_0$). The rotating examples resolve this by conjugating with a
group element ($e^{2\pi J}=P$ in Example~\ref{ex:holonomy},
$U_\pi=P$ in Remark~\ref{rem:oscillator}) that commutes with $A_0$
--- the general mechanism behind
Theorem~\ref{thm:classification}\emph{(ii)}: monodromy must lie in
the commutant.
\end{remark}

\section{Relation to the Discrete Theory}
\label{sec:discrete}

The companion paper \cite{alexa2026} developed the theory for a finite
index set $I$. We now make the embedding of that theory into the
present framework precise: the discrete theory corresponds exactly to
the flat sector of the continuous theory, the case in which the base
space carries no non-contractible loops and holonomy is automatically
trivial.

\begin{proposition}[Discrete reduction]
\label{prop:discrete-reduction}
Let $I$ be a finite set equipped with the discrete topology, regarded
as a $0$-dimensional manifold. Then:

\noindent\emph{(a)} If the operators share a common dense domain,
every family $\{A_i\}_{i\in I}$ of positive self-adjoint operators
with compact resolvent on $\mathcal{H}$ is a continuous operator
network over $I$ literally in the sense of
Definition~\ref{def:network}, all continuity requirements being
vacuous for the discrete topology. Without a common domain, the
results of Sections~\ref{sec:cocycles}--\ref{sec:gauge} ---
diagonal invertibility, generator transport, multiplicativity, and
gauge rigidity --- remain valid verbatim with individual domains
$\mathcal{D}(A_i)$, since no relation between domains at different
points is used there.

\noindent\emph{(b)} Every two-point intertwining cocycle over $I$
(Definition~\ref{def:cocycle}) coincides with a time-scaled
intertwining cocycle in the sense of \cite{alexa2026}. The gauge
rigidity Theorem~\ref{thm:continuous-gauge} reduces to the gauge
rigidity theorem of that paper: $\lambda_{ij}=\tau_i/\tau_j$ for
positive constants $\{\tau_i\}_{i\in I}$, unique up to a single
multiplicative constant --- since $\lambda$ is defined on all pairs,
the uniqueness argument of Theorem~\ref{thm:continuous-gauge} is
purely algebraic and needs no connectedness.

\noindent\emph{(c)} Over $I$ there is no monodromy of any kind. As a
$0$-dimensional manifold, $I$ admits no non-constant continuous
paths, so every path cocycle over $I$ has $K_\gamma=I$ for all loops.
The combinatorial cycles of the companion paper --- finite sequences
$i_0\to i_1\to\cdots\to i_k=i_0$ --- are compositions of two-point
transports and collapse to the identity by the cocycle
identity~\eqref{eq:cocycle-C}.
\end{proposition}

\begin{proof}
Parts (a) and (b) are immediate from the definitions; the discrete
topology makes all continuity conditions vacuous, and the
finite-index gauge rigidity is a special case of
Theorem~\ref{thm:continuous-gauge}: $\tau_i=\lambda(i,i_0)$ for any
fixed base index $i_0$, by multiplicativity.

For part (c): a continuous map $\gamma\colon[0,1]\to I$ from a
connected interval to a discrete space is constant, so $K_\gamma=I$ by
\emph{(P1)}. For the combinatorial cycles, the cocycle
identity~\eqref{eq:cocycle-C} applied repeatedly gives
$K_{i_0 i_{k-1}}\cdots K_{i_1 i_0}=K_{i_0 i_0}=I$ by
Proposition~\ref{prop:diagonal}(i). In this precise sense the
discrete theory is the two-point (groupoid) analogue of the
endpoint-dependent flat sector of the path theory, rather than a
literal special case of it: a discrete base supports no non-constant
paths at all.
\end{proof}

\begin{remark}
The content of Proposition~\ref{prop:discrete-reduction}(c) is that
nontrivial holonomy is genuinely a phenomenon of the continuous
parameter space. When $M$ has a non-contractible loop --- as in the
example of Section~\ref{sec:example} where $M = S^1$ --- the holonomy
group can be nontrivial even for a flat connection ($F=0$), precisely
because $\pi_1(S^1)\cong\mathbb{Z}$.
Theorem~\ref{thm:classification} describes exactly which monodromies
occur. The discrete theory of \cite{alexa2026} is thus the degenerate
case in which this topological obstruction is absent.
\end{remark}

\section{Conclusion}
\label{sec:conclusion}

The central finding of this paper is that gauge rigidity is not an
artifact of discrete index sets but a structural consequence of the
semigroup intertwining relation in full generality. Whenever a
continuous family of dissipative (or unitary) dynamics is linked by a
cocycle of time-scaled transfer maps, the scaling field is necessarily
of gauge form $\lambda(x,y)=\tau(x)/\tau(y)$. The key step --- that
multiplicativity of $\lambda$ follows from the semigroup structure
rather than being imposed --- is the same algebraic argument that
works in the discrete case, and it works identically here.

What is genuinely new in the continuous setting is the geometry that
emerges once the index set becomes a manifold. The transfer operators
define a connection on a Hilbert bundle, and transport around
non-contractible loops acquires monodromy. This monodromy is tightly
constrained: it commutes with the generator at the base point,
decomposes into spectral sectors carrying adiabatic transport twisted
by independent commutant-valued sector potentials, and, for flat
unitary cocycles, is classified by unitary representations of
$\pi_1(M)$ in the commutant. Nontrivial holonomy is a concrete phenomenon: for a
rotating network on $\ell^2$ over $S^1$ (with a rotating anisotropic
oscillator as its physical counterpart), the holonomy is the parity
operator --- invisible to the gauge function $\tau$ and to the Berry
one-form of every level. The discrete theory of \cite{alexa2026} is
recovered as the two-point analogue of the flat, topologically
trivial sector of the present framework.

The relation to adiabatic quantum mechanics is now precise: the
restriction of the transport to each renormalized spectral sector is
Berry--Wilczek--Zee transport twisted by the sector potential, so
adiabatic theory captures the transport exactly on the class of
adiabatically normalized cocycles. The gap between holonomy and the
local Berry data is witnessed twice: by the sector potential, and,
even in the normalized case, by global topology --- in the example
every Berry one-form vanishes identically, yet the holonomy is a
nontrivial $\mathbb{Z}_2$ monodromy carried by M\"{o}bius eigenline
bundles.

Several questions remain open. The realization part of the
classification (Theorem~\ref{thm:classification}\emph{(ii)})
produces holonomies canonically on flat Hilbert bundles and, under
the stated triviality condition, on constant networks; classifying
flat cocycles over a \emph{fixed} network is a finer problem. The inverse
problem of reconstructing $\tau$ from spectral data of a continuous
mixture semigroup has not been addressed in the continuous setting.
The stability of holonomy under perturbations of $\{A_x\}$ is relevant
to robustness questions in adiabatic theory. Extensions to
non-self-adjoint generators and to infinite-dimensional parameter
spaces $M$ are natural next steps.

\section*{Statements and Declarations}

\noindent\textbf{Author contributions.}
A.A.\ performed all aspects of the work, including conceptualization,
mathematical analysis, and manuscript preparation.

\medskip
\noindent\textbf{Funding.}
The author received no financial support for the research, authorship,
or publication of this article.

\medskip
\noindent\textbf{Data availability.}
No data were generated or analyzed in this work.

\medskip
\noindent\textbf{Competing interests.}
The author declares no competing interests.

\bibliographystyle{amsplain}
\bibliography{references}

\end{document}